\documentclass[11pt,letterpaper]{article}
\usepackage[T1]{fontenc}
\usepackage{graphicx}

\usepackage{geometry}
\usepackage{amsmath,amssymb,amsthm}
\usepackage{hyperref}
\usepackage{xcolor}

\theoremstyle{plain}
\newtheorem{theorem}{Theorem}
\newtheorem{proposition}{Proposition}
\newtheorem{lemma}{Lemma}
\newtheorem{corollary}{Corollary}
\theoremstyle{definition}
\newtheorem{definition}{Definition}
\newtheorem{remark}{Remark}

\usepackage{titling}
\usepackage[noblocks]{authblk}

\usepackage[square,numbers]{natbib}
\usepackage{microtype}
\usepackage[tt=false]{libertine}
\usepackage{bbm}

\begin{document}
	
	\title{Revisiting the Fairness of Pure Nash Equilibria of Round-Robin}

	\author[1,2]{Georgios Amanatidis}
	\affil[1]{Department of Informatics, Athens University of Economics and Business, Greece.}
	\affil[2]{Archimedes / Athena RC, Greece.}
	
	\author[3]{Haris Aziz}
	\affil[3]{School of Computer Science and Engineering, UNSW Sydney, Australia.}
	
	\author[4]{Georgios Birmpas}
	\affil[4]{Department of Computer Science, University of Liverpool, UK.}
	
	\author[5]{Federico Fusco}
	\affil[5]{Department of Computer, Control, and Management Engineering ``Antonio Ruberti'', Sapienza University of Rome,  Italy.}
	
	\author[5]{Stefano Leonardi}
	
	\author[6]{Rebecca Reiffenh{\"a}user}
	\affil[6]{Institute for Logic, Language and Computation, University of Amsterdam, The Netherlands.  }

	\predate{}
	\postdate{}
	\date{}

	\maketitle

	\begin{abstract}
We prove that for additive valuations, not every pure Nash equilibrium of the Round-Robin mechanism induces an allocation that is envy-free up to one good (EF1) with respect to the agents’ true valuations. This corrects an earlier claim of Amanatidis et al. (WINE 2021; Math. Oper. Res. 2024). On the positive side, every pure Nash equilibrium of every fixed recursively balanced picking sequence mechanism induces an allocation that is $1/2$-EF1 and proportional up to one good (PROP1). We prove the approximate EF1 guarantee through a potential function comparing two executions, and the proportionality guarantee via a deviation that preserves the agent’s most valuable goods. We also give a family of Round Robin equilibria showing that it is not possible to have a strictly better guarantee than $2/3$-EF1 for all equilibria.
	\end{abstract}
	

	\section{Introduction}
	
	In the fair division of indivisible goods, envy-freeness up to one good
	(EF1) is one of the central fairness notions, and the Round-Robin
	algorithm~\citep{BoLa11a}, in which agents take turns picking their favourite remaining
	good, is the simplest way to achieve it: when agents report their
	valuations truthfully, Round-Robin always returns an EF1
	allocation~\cite{CKMPSW19}.
	
	Strategic agents, however, may misreport (see, e.g., \citep{KoCh71a,ABLM17a}). \citet{ABFLLR21,ABFLLR24} claim that, for additive nonnegative valuations, every pure Nash equilibrium of the Round-Robin reporting mechanism still induces an allocation that is EF1 with respect to the true valuations~\citep[Theorem~4]{ABFLLR24}.
	We give a counterexample under the reporting model stated in their paper.
	
	Despite the failure of exact EF1, we prove that two guarantees survive: every pure Nash equilibrium
	is $1/2$-EF1 and proportional up to one good (PROP1). Both statements
	hold for the larger class of fixed recursively balanced picking
	sequences~\citep{AWX15b}, in which the order of the agents may change between rounds.
	The $1/2$-EF1 guarantee is already known for Round-Robin as a special
	case of the result of Amanatidis et al. for subadditive cancelable
	valuations~\citep[Theorem~8]{ABLLR25}. We give a direct proof
	for additive valuations and extend the guarantee to every fixed
	recursively balanced sequence. The existence of a pure Nash equilibrium for any recursively balanced picking sequence mechanism follows from the bluff profile of \citet{AGW17}; see the discussion at the end of Section~\ref{sec:prop1}.
	Finally, we give a more elaborate family of Round-Robin equilibria
	showing that no factor strictly greater than $2/3$ can be guaranteed with respect to approximate EF1 for all equilibria.

	\subsection{Our contributions}
	
	We first show that exact EF1 need not hold at equilibrium, even with
	three agents and eight goods (Proposition~\ref{prop:non-ef1}).
	In our example, two agents value all
	goods equally, and the third uses only three distinct values. All
	reports are strict, so the example is independent of any tie-breaking rule used. 
	
	Our main positive result (Theorem~\ref{thm:guarantees}) combines
	$1/2$-EF1 and PROP1 for every fixed recursively balanced sequence.
	The pairwise guarantee follows from a
	potential function that compares the equilibrium execution with a deviation
	targeting another agent's bundle. It requires only that the agent
	whose envy we examine best responds to the other reports. The proportionality guarantee uses a different deviation, in which the agent first takes all goods that she values more than every good she does not receive, and then picks truthfully. Again, this only requires that her report is a best response.
	
	Finally, Theorem~\ref{thm:upper-bound} gives a family of Round-Robin
	equilibria ruling out every universal EF1 approximation factor
	greater than $2/3$. All agents except the one whose envy we examine have identical unit valuations, and the latter uses only three distinct values. The main step is to show that the agent whose envy we examine best responds, and this reduces to a counting argument: for every two goods that she takes from the bundle she envies, she loses at least three of her own goods.

	\section{Setting}\label{sec:setting}
	
	An instance consists of a set $N=\{1,\dots,n\}$ of agents, a set $M$ of
	indivisible goods with $m=|M|$, and, for each agent $i\in N$, an additive valuation
	$v_i$ with $v_i(g)\ge 0$ for every $g\in M$ and
	$v_i(S)=\sum_{g\in S}v_i(g)$ for every $S\subseteq M$. An
	\emph{allocation} $(A_1,\dots,A_n)$ is a partition of $M$, where $A_i$
	is the bundle of agent~$i$.
	
	
	\begin{definition}[$\alpha$-EF1]
		For $\alpha\in[0,1]$, an allocation $(A_1,\dots,A_n)$ is
		\emph{$\alpha$-EF1} if, for every pair $i,j\in N$ with
		$A_j\ne\emptyset$, there is a good $g\in A_j$ such that
		\[
		v_i(A_i)\ge\alpha\,v_i(A_j\setminus\{g\}).
		\]
		When $\alpha = 1$ we just say that the allocation is EF1.
	\end{definition}
	
	\begin{definition}[PROP1]
		An allocation $(A_1,\dots,A_n)$ is \emph{proportional up to one good} if,
		for every agent $i$, either $v_i(A_i)\ge v_i(M)/n$ or there is a good
		$g\in M\setminus A_i$ such that
		$v_i(A_i\cup\{g\})\ge v_i(M)/n$.
	\end{definition}
	
	For nonnegative additive valuations, an EF1 allocation is also PROP1. 
	
	\paragraph{Sequential mechanisms.}
	Each agent $i$ reports a bid vector $b_i\in\mathbb{R}^{M}_{\ge 0}$,
	which need not coincide with $v_i$. A picking sequence, fixed
	independently of the reports, specifies the active agent at each turn, and the active agent receives the available
	good with her highest reported bid; ties are broken lexicographically. The reports are submitted once, before the allocation begins, and remain fixed throughout the execution.
	A picking sequence is \emph{recursively balanced} if it consists of
	complete rounds, each of which is a permutation of $N$, followed
	possibly by an incomplete round in which no agent appears twice.
	Round-Robin repeats the order $(1,2,\dots,n)$, truncating the last
	round when all goods have been allocated.

	Under truthful reports, a recursively balanced sequence produces an EF1 allocation: for every $r$, the good that agent~$i$ picks at her $r$th turn is at least as valuable to her as the good that any other agent~$j$ picks at the $(r+1)$th turn of~$j$; so, removing the first good of $j$'s bundle eliminates the envy of agent~$i$ towards~$j$. Hence, it also produces a PROP1 allocation.
    For Round-Robin, see also~\citep{CKMPSW19}.

	Since the outcome depends on a report only through the strict ranking of goods that
	it induces after tie-breaking, and every strict ranking is induced by
	some bid vector, we identify each report with a strict ranking
	$\widehat{\succ}_i$ of $M$, as in~\cite{ABFLLR24}.
	
	\begin{definition}[Pure Nash equilibrium]
		A report profile $(\widehat{\succ}_1,\dots,\widehat{\succ}_n)$ is a
		\emph{pure Nash equilibrium} if no agent $i$ can obtain a bundle of
		strictly higher true value $v_i$ by unilaterally changing her report.
	\end{definition}
	
	\section{Not every pure Nash equilibrium induces an EF1 allocation}\label{sec:example}

	\begin{proposition}\label{prop:non-ef1}
		There is a three-agent, eight-good additive instance with strictly
		positive valuations and a pure Nash equilibrium of the Round-Robin
		reporting mechanism whose outcome is not EF1.
	\end{proposition}
	
	\begin{proof}
		Let $M=\{a,b,c,x,y,z,p,q\}$, and let the picking sequence be
		$(1,2,3\mid 1,2,3\mid 1,2)$. Agents~1 and~3 value every good at~$1$.
		Agent~2's values are
		\[
		\begin{array}{c|cccccccc}
			&a&b&c&x&y&z&p&q\\ \hline
			v_2&8&8&8&5&5&5&2&2
		\end{array}
		\]
		and the reported rankings are
		\[
		\begin{aligned}
			\widehat{\succ}_1&:\quad
			a\succ b\succ c\succ x\succ y\succ z\succ p\succ q,\\
			\widehat{\succ}_2&:\quad
			x\succ y\succ z\succ a\succ b\succ c\succ p\succ q,\\
			\widehat{\succ}_3&:\quad
			p\succ x\succ y\succ q\succ z\succ a\succ b\succ c.
		\end{aligned}
		\]
		These are strict rankings, and can be implemented by strictly decreasing
		bids. Thus no tie-breaking is needed. The goods are picked in the order
		$a,x,p,b,y,q,c,z$, giving
		\[
		A_1=\{a,b,c\},\qquad A_2=\{x,y,z\},\qquad A_3=\{p,q\}.
		\]
		
		Agents~1 and~3 always receive three and two goods, respectively, so
		neither can improve her utility. Agent~2 receives value~$15$. To verify
		that she also has no profitable deviation, first observe that $a$ is
		always taken immediately. Moreover, she can receive at most one of $b,c$, and only on her first or second turn. Indeed, agent~3 never takes $b$ or $c$, since she ranks them last and at least three goods are available at each of her turns. So, if agent~2 takes one of $b,c$ on her first turn, agent~1 takes the other on her second turn; otherwise, agent~1 takes $b$ on her second turn and $c$ on her third turn, unless agent~2 has already taken it.

		Without $b,c$, agent~2 can obtain value at most~$15$. With one of $b,c$
		and at most one of $x,y,z$, her value is at most $8+5+2=15$.
		It remains to exclude a bundle containing one of $b,c$ and two of
		$x,y,z$. For such a bundle, her first two picks must consist of one
		good from each of these two sets. At that point, agent~1 has taken $a$
		and the other good from $\{b,c\}$, and agent~3 has taken $p$.
		Only $q$ and two goods from $\{x,y,z\}$ remain. At least one of these
		two goods is $x$ or $y$, which agent~3 ranks above $q$, so agent~3 takes
		one of them on her second turn. Agent~1 ranks the other above $q$ and
		takes it on her third turn. Agent~2 must therefore take $q$ on her
		third turn, a contradiction. This proves that the profile is a pure
		Nash equilibrium.
		
		Finally, $v_2(A_2)=15$, whereas
		$v_2(A_1\setminus\{g\})=16$ for every $g\in A_1$. Thus the outcome
		is not EF1.
	\end{proof}
	
	\begin{remark}
		The gap in the argument of \citet{ABFLLR24} (which also appears in~\citep{ABFLLR21}) occurs in the proof of their Lemma~2, at inequality~(2). Their recursive construction replaces $v_1$ by intermediate valuation functions $v_1^{r+1}$, but the best-response condition invoked at that point holds only for the original valuation function $v_1$. Preserving the total value of $A_1$ and the values of all goods outside $A_1$, as the construction does, does not preserve the values of deviating bundles that contain only some of the goods of $A_1$.
	\end{remark}

	\section{Equilibrium guarantees for recursively balanced sequences}
	\label{sec:prop1}
	
	Although EF1 can fail at equilibrium, every pure Nash equilibrium is still $1/2$-EF1 and PROP1. This remains true even when the order of the agents changes between rounds.
	
	\begin{theorem}\label{thm:guarantees}
		For nonnegative additive valuations, every pure Nash equilibrium of
		every fixed recursively balanced picking sequence mechanism induces
		an allocation that is $1/2$-EF1 and PROP1 with respect to the agents'
		true valuation functions.
	\end{theorem}
	
	We prove the two guarantees separately. The first proposition gives
	a quantitative proportionality bound. The second gives the pairwise
	bundle value comparison needed for $1/2$-EF1.
	
	\begin{proposition}\label{prop:prop1}
		Let $A=(A_1,\dots,A_n)$ be the allocation induced by a report profile for a fixed recursively balanced picking
		sequence in which agent $i$'s report is a best response
		to the other reports. Let
		$
		h_i=\max_{g\in M\setminus A_i}v_i(g),
		$
		where the maximum is defined as zero if $M\setminus A_i$ is empty. Then
		$
		v_i(A_i)+\frac{n-1}{n}h_i\ge \frac{v_i(M)}{n}.
		$
	\end{proposition}

	A natural approach would be to compare $A_i$ with the bundle that agent~$i$ would get by reporting truthfully, as in the standard proof that recursively balanced picking sequences produce EF1 allocations under truthful reports. However, a truthful report may cost agent~$i$ some of her most valuable goods. Instead, we let her first secure all goods that she values more than every good she does not receive in $A_i$, and only then pick truthfully. The next lemma shows that such a deviation is always possible.
	
	\begin{lemma}\label{lem:promotion}
		Fix a picking sequence and the reports of the other agents. Let $g_1,\dots,g_c$ be any distinct goods that
		agent~$i$ receives under report $\widehat{\succ}_i$, listed in the order she acquires
		them, and let $\widehat{\succ}'_i$ be any report ranking $g_1,\dots,g_c$ first, in
		this order. Then under $\widehat{\succ}'_i$ agent~$i$ receives $g_r$ at her $r$th
		turn, for every $r\le c$.
	\end{lemma}

	\begin{proof}
		If $c=0$, the statement is vacuous. Suppose $c\ge1$.
		Let $S_t$ and $S'_t$ be the sets of goods allocated before turn $t$ under
		$\widehat{\succ}_i$ and $\widehat{\succ}'_i$. We show by induction on $t$, up to agent~$i$'s
		$c$th turn, that she has received $g_1,\dots,g_{r-1}$ at her turns before
		$t$ and that $S'_t\subseteq S_t\cup\{g_1,\dots,g_{r-1}\}$, where $r-1$ is
		the number of her turns before $t$. Initially $S_1=S'_1=\emptyset$,
		so the invariant holds. Assume it holds immediately before turn $t$.
		
		If turn $t$ is agent~$i$'s, then $g_r\notin S_t$, since under $\widehat{\succ}_i$
		she acquires $g_r$ at her $r$th turn or later, and
		$g_r\notin\{g_1,\dots,g_{r-1}\}$. So $g_r$ is available, and as
		$g_1,\dots,g_{r-1}$ are gone, it is the highest-ranked available good of
		$\widehat{\succ}'_i$; she takes it. Adding $g_r$ to $S'_t$ and to the list
		of promoted goods preserves the inclusion at the next turn.
		If turn $t$ belongs to $j\ne i$, let $x$ be
		$j$'s pick under $\widehat{\succ}_i$. Then $x\notin S_t$ and
		$x\notin\{g_1,\dots,g_{r-1}\}$, since these goods are allocated to
		agent~$i$ under $\widehat{\succ}_i$. The induction hypothesis therefore implies
		that $x$ is available under $\widehat{\succ}'_i$. Hence $j$ takes $x$ or a good
		ranked above $x$. Any such higher-ranked good belongs to $S_t$, since
		otherwise $j$ would have taken it under $\widehat{\succ}_i$. Thus the new pick
		belongs to $S_{t+1}$, preserving the invariant.
	\end{proof}
	
	We now prove the quantitative bound. 
	
	\begin{proof}[Proof of Proposition~\ref{prop:prop1}]
		Let $H=\{g\in M: v_i(g)>h_i\}$. By the definition of $h_i$,
		$H\subseteq A_i$, and every good outside $H$ is worth at most $h_i$ to
		agent~$i$. Let
		$g_1,\dots,g_c$ be the goods of $H$ in the order agent~$i$ acquires them.
		Recursive balance gives two facts: every agent's $r$th turn lies in round
		$r$, and the numbers of turns of any two agents differ by at most one.
		
		To exploit the best-response condition, we consider a 
		unilateral deviation by agent~$i$, while keeping all other reports fixed.
		We construct this deviation so that, denoting its resulting allocation by
		$(D_1,\dots,D_n)$, all but one good of every $D_j$ can be paired with
		distinct goods of $D_i$ that agent~$i$ values at least as much. The
		remaining good is worth at most $h_i$. This gives a lower bound on
		$v_i(D_i)$; because the original report is a best response, agent~$i$
		values her actual bundle $A_i$ at least as much as $D_i$.

		Consider the report $\widehat{\succ}'_i$ ranking $g_1,\dots,g_c$ first, in this
		order, followed by the remaining goods in nonincreasing order of $v_i$,
		breaking ties lexicographically, and
		denote the resulting allocation by $(D_1,\dots,D_n)$. By
		Lemma~\ref{lem:promotion}, agent~$i$ receives $g_r$ at her $r$th turn
		for $r\le c$; at every later turn all of $H$ is gone, so she takes a
		most valuable available good. Let $d_r$ be her $r$th pick.
		
		Fix $j\ne i$. We pair every good of $D_j$ except its first with an
		earlier pick of agent~$i$. If $D_j=\emptyset$, then
		$v_i(D_j)\le v_i(D_i)+h_i$ is immediate. Otherwise, list the goods of
		$D_j$ in allocation order as $y_1,\dots,y_q$, where
		$q\le|D_i|+1$. We claim $v_i(d_r)\ge v_i(y_{r+1})$ for
		$1\le r\le q-1$. If $r\le c$, then $d_r=g_r$ is worth more than $h_i$,
		while $y_{r+1}\notin H$ is worth at most $h_i$. If $r>c$, then $y_{r+1}$ is
		taken in round $r+1$, hence is still available at agent~$i$'s $r$th
		turn, where she takes a most valuable available good. Therefore
		$v_i(d_r)\ge v_i(y_{r+1})$ in this case as well.
		
		Since $y_1\notin H$, we have $v_i(y_1)\le h_i$, and thus
		\[
		\begin{aligned}
			v_i(D_j)
			&= v_i(y_1)+\sum_{r=1}^{q-1}v_i(y_{r+1})\\
			&\le h_i+\sum_{r=1}^{q-1}v_i(d_r)\\
			&\le h_i+v_i(D_i).
		\end{aligned}
		\]
		Here, the first inequality uses
		$v_i(y_{r+1})\le v_i(d_r)$ for every $r\le q-1$. For the second
		inequality, recall that $q-1\le |D_i|$, so
		$d_1,\dots,d_{q-1}$ form a subset of $D_i$; the inequality therefore
		follows from the nonnegativity of the valuations.

		Since
		$(D_1,\dots,D_n)$ partitions $M$,
		\[
		v_i(M)
		=v_i(D_i)+\sum_{j\ne i}v_i(D_j)
		\le n\,v_i(D_i)+(n-1)h_i.
		\]
		
		Hence, 
		\[v_i(D_i)+\frac{n-1}{n}h_i\ge v_i(M)/n \,.\]
		Finally, $v_i(A_i)\geq v_i(D_i)$ because the original report is a best
		response. Hence, 	\[v_i(A_i)+ \frac{n-1}{n}h_i\geq  v_i(M)/n \,.\qedhere\]
	\end{proof}

	For the pairwise guarantee, we show $\tfrac 12$-EF1 by considering the deviation in which agent~$i$ ranks the goods of $A_j$ first, from the most to the least valuable to her. The difficulty is that this deviation may change the picks of the other agents, and such changes may propagate until the end of the execution. To control them, we compare the two executions turn by turn, using a potential function that can increase only at the turns of agents~$i$ and~$j$. Roughly speaking, each good of $A_j$ that agent~$i$ takes in the deviation pays for at most two turns of agent~$j$ in the original execution, while the first turn of agent~$j$ is paid for by a most valuable good of $A_j$, i.e., by the good that is removed in the EF1 comparison. This is where the factor $1/2$ comes from.

	\begin{proposition}\label{prop:half-ef1}
		Fix a recursively balanced picking sequence independently of the
		reports. If agent $i$'s report is a best response to the other reports
		and the resulting allocation is $A=(A_1,\dots,A_n)$, then, for every
		$j\ne i$ with $A_j\ne\emptyset$,
		\[
		v_i(A_i)\ge\frac12\Big(v_i(A_j)-\max_{g\in A_j}v_i(g)\Big).
		\]
	\end{proposition}
	\begin{proof}
		Fix $j\ne i$ with $A_j\ne\emptyset$, and set
		$h=\max_{g\in A_j}v_i(g)$. Suppose that agent $i$ deviates to a strict ranking
		$\widehat{\succ}'_i$ that places all goods of $A_j$ first in
		nonincreasing order of $v_i$, breaking ties lexicographically, followed
		by the goods outside $A_j$ in any strict order.
		We denote the deviating allocation by $D=(D_1,\dots,D_n)$.

		Let $k$ be the number of turns of $i$, and let $d_r$ be her deviating
		pick on her $r$th turn, for $1\le r\le k$. Define
		\[
		x_0=h,\qquad
		x_r=\begin{cases}v_i(d_r),&d_r\in A_j,\\0,&d_r\notin A_j.\end{cases}
		\]
		Whenever a good of $A_j$ is available, $i$ takes a most valuable such
		good. If she takes a good outside $A_j$, no good of $A_j$ remains.
		Since available sets only shrink, it follows that
		\[
		x_0\ge x_1\ge\cdots\ge x_k\ge0.
		\]
		After her $r$th turn, every remaining good of $A_j$ has adjusted value at most
		$x_r$. Additivity and nonnegativity also give
		\begin{equation}\label{eq:targetdeviationvalue}
			v_i(D_i)=\sum_{r=1}^k v_i(d_r)\ge\sum_{r=1}^k x_r.
		\end{equation}
		These statements also apply when $k=0$, with empty sums.
		
		Both executions use the same picking sequence and allocate one good
		at each turn. We synchronize them and write $R,S$ for the respective
		sets of remaining goods in the original and deviating executions,
		respectively. For a nonnegative \textit{level} $\lambda$, define
		\begin{equation}\label{eq:localpotential}
			\Phi(R,S,\lambda)
			=v_i(A_j\setminus R)+
			\sum_{g\in R\setminus S}\max\{\lambda, \mathbbm{1}_{A_j}(g)\cdot v_i(g) \}.
		\end{equation}
		The first term records the value allocated to $j$ in the original
		execution. The second term assigns a \emph{residual value} to every good that is still available in the original execution but not in the deviating one. The level $\lambda$ ensures that each such good has residual value at least $\lambda$, even if it is not a good of $A_j$.

		Initially set $\lambda=x_0=h$. Immediately before $i$'s $r$th turn,
		change the level to $x_r$, and retain this level until her next turn.
		Every level change is a decrease, so it does not  increase $\Phi$ right at that moment. 
		Initially $R=S=M$, hence $\Phi=0$. Below, we bound the change of $\Phi$ at each turn, for the level that is in effect at that turn. \medskip

		\noindent\emph{A turn of agent $i$:}
		Let $a_r$ be her original pick. Since $j\ne i$, we have
		$a_r\notin A_j$, so the first term of the potential does not change.
		The new residual value set is
		\[
		((R\setminus S)\setminus\{a_r\})
		\cup (\{d_r\}\cap(R\setminus\{a_r\})).
		\]
		Thus a residual value term may be removed, and at most one is added, for $d_r$ (only if it remains available in the original execution, i.e., if $d_r\in R\setminus\{a_r\}$).
		Because $\mathbbm{1}_{A_j}(d_r)\cdot v_i(d_r)=x_r=\lambda$, the added residual value equals $x_r$.
		This turn therefore increases $\Phi$ by at most $x_r$. \medskip
		
		\noindent\emph{A turn of an agent $\ell\ne i$:}
		Let $u\in R$ and $z\in S$ be her original and deviating picks.
		Every  good $g\in R\cap S$ satisfies $\mathbbm{1}_{A_j}(g)\cdot v_i(g)\le\lambda$.
		Before $i$'s first turn this follows from $\lambda=h$; subsequently,
		it follows from the bound on goods remaining after her most recent
		deviating pick.
		
		If $u\ne z$, the two goods cannot both be in $R\cap S$; otherwise the
		unchanged strict ranking $\widehat{\succ}_\ell$ would have to rank
		$u$ above $z$ in the original execution and $z$ above $u$ in the
		deviating execution. 
		So, if $u\in R\cap S$, then either $u=z$ or $z\in S\setminus R$; in both cases the set $R\setminus S$ does not change. If $u\in R\setminus S$, then the residual value of $u$ is removed, and $z\ne u$, as $u\notin S$. In this case, no residual value is added if $z\in S\setminus R$, whereas a residual value of $\max\{\mathbbm{1}_{A_j}(z)\cdot v_i(z),\lambda\}=\lambda$ is added if $z\in R\cap S$.
		
		If $\ell\ne j$, then $u\notin A_j$. The first term of $\Phi$ is unchanged.
		In the case where $u \in R\cap S$, the residual values are unchanged, whereas in the
		case where $u\in R\setminus S$, a residual value of $\lambda$ is removed and at most a residual value of $\lambda$ is added. Thus this turn does not increase $\Phi$.
		
		If $\ell=j$, then $u\in A_j$ and the first term increases by
		$v_i(u)$. 
		When $u \in R\cap S$, this is at most $\lambda$, and
		residual values are unchanged. Otherwise (i.e., when $u\in R\setminus S$), by the discussion above, the total increase is at most
		\[
		v_i(u)-\max\{v_i(u),\lambda\}+\lambda\le\lambda.
		\]
		Hence, every turn of $j$ increases $\Phi$ by at most its current level. \medskip

		Now we may bound the final value of  $\Phi$ by summing over turns.
		Let $q$ be the total number of turns of $j$; by recursive balance, we have $q\le k+1$. Before $j$'s $s$th turn, suppose that $i$ has taken $p$ turns; again, by recursive balance, we have  $p\ge s-1$.
		The current level is $x_p$, and therefore
		\[
		\lambda=x_p\le x_{s-1}.
		\]
		Note that this includes $s=1$, as $x_0=h$.
		
		At termination $R=S=\emptyset$, so the final potential is $v_i(A_j)$.
		Summing the turn increases, including the nonpositive level changes,
		and then using $q-1\le k$, we obtain
		\begin{align*}
			v_i(A_j) &= \Phi(\emptyset, \emptyset, x_k)\\
			&\le \Phi(M, M, x_0) + \sum_{r=1}^k x_r+\sum_{s=1}^q x_{s-1}\\
			&=\sum_{r=1}^k x_r+h+\sum_{r=1}^{q-1}x_r \le h+2\sum_{r=1}^k x_r\\
			&\le h+2v_i(D_i) \le h+2v_i(A_i).
		\end{align*}
		The second to last inequality is \eqref{eq:targetdeviationvalue}, whereas
		the last one uses the best-response condition for $\widehat{\succ}'_i$.
		Rearranging the terms gives the claimed inequality.
	\end{proof}
	
	We are now ready to prove Theorem~\ref{thm:guarantees} using the two propositions above.

	\begin{proof}[Proof of Theorem~\ref{thm:guarantees}]
		Let $A$ be an allocation induced by a pure Nash equilibrium. Every agent is a best responder,
		so Proposition~\ref{prop:half-ef1} applies to every ordered pair of
		distinct agents. For $A_j\ne\emptyset$, choose a good $g\in A_j$
		maximizing $v_i(g)$. By additivity,
		\[
		v_i(A_i)\ge\frac12\left(v_i(A_j)-v_i(g)\right)
		=\frac12v_i(A_j\setminus\{g\}).
		\]
		Thus $A$ is $1/2$-EF1.
		
		For PROP1, fix $i$ and use $h_i$ from Proposition~\ref{prop:prop1}, i.e., $h_i=\max_{g\in M\setminus A_i}v_i(g)$.
		If $v_i(A_i)\ge v_i(M)/n$, proportionality already holds; otherwise,
		$M\setminus A_i$ must be nonempty. Consider a good $g\in M\setminus A_i$
		with $v_i(g)=h_i$. Proposition~\ref{prop:prop1} and nonnegativity give
		\[
		v_i(A_i\cup\{g\})=v_i(A_i)+h_i
		\ge v_i(A_i)+\frac{n-1}{n}h_i
		\ge\frac{v_i(M)}n.
		\]
		Thus $A$ is also PROP1.
	\end{proof}

	\begin{corollary}\label{cor:two-agents}
		For two agents with nonnegative additive valuations, every pure Nash
		equilibrium of every fixed recursively balanced picking sequence mechanism
		induces an EF1 allocation.
	\end{corollary}
	\begin{proof}
		For $j\ne i$, Proposition~\ref{prop:prop1} with $n=2$ gives
		$2v_i(A_i)+h_i\ge v_i(M)=v_i(A_i)+v_i(A_j)$.
		Since $M\setminus A_i=A_j$, choosing a most valuable good of
		a nonempty $A_j$ proves EF1.
	\end{proof}
	For Round-Robin, this was already known from~\citep[Theorem~3]{ABFLLR24}.

	\paragraph{Approximate equilibria.}
	The proofs above also give bounds when best responses are approximate, using the exact same analysis.
	For $\beta\in(0,1]$, suppose that the report profile is a \emph{$\beta$-approximate pure Nash equilibrium}, i.e., every agent $i$ satisfies
	$v_i(A_i)\ge\beta v_i(D_i)$ for every unilateral deviation with
	resulting allocation $D$. The same deviations then give
	\[
	v_i(A_i)\ge\frac{\beta}{2}
	\Big(v_i(A_j)-\max_{g\in A_j}v_i(g)\Big)\,,
	\quad\text{for }j\ne i\text{ with }A_j\ne\emptyset,
	\]
	and
	\[
	v_i(A_i)+\beta\frac{n-1}{n}h_i\ge\beta\frac{v_i(M)}n.
	\]
	Thus the allocation is $\beta/2$-EF1, and each agent can reach
	at least $\beta v_i(M)/n$ by adding at most one outside good.
	For Round-Robin, the pairwise guarantee is already known due to~\citep[Theorem~8]{ABLLR25}.\medskip

	Our counterexample separates exact EF1 from the two equilibrium
	guarantees of Theorem~\ref{thm:guarantees}. EF1 may fail, whereas $1/2$-EF1 and PROP1 hold for every
	pure Nash equilibrium of every fixed recursively balanced mechanism.
	There also always exists an EF1 pure Nash equilibrium: the bluff
	profile reproduces the truthful allocation~\citep[Theorem~2]{AGW17}.%

	Theorem~2 of \citet{AGW17} is stated for positive values that induce strict rankings, but it extends to our setting. Refine true-value ties lexicographically and let $r_i(g)$
	be a good's rank, with rank~1 denoting the most valuable good.
	For $\delta>0$, replace $v_i(g)$ by
	$v_i(g)+\delta(m+1-r_i(g))$. The true strict rankings, truthful
	allocation, and bluff profile are unchanged. The profile is a pure Nash equilibrium
	for every $\delta>0$; taking $\delta\downarrow0$ in each best-response
	inequality proves that it is a pure Nash equilibrium for the original values as well. For Round-Robin, this was also shown in~\citep[Theorem~10]{ABFLLR24}.

	\section{An upper bound of \texorpdfstring{$2/3$}{2/3} for Round-Robin}
	\label{sec:upper-bound}

	The counterexample in Proposition~\ref{prop:non-ef1} already gives
	an upper bound of $15/16$: agent~2 receives value $15$, while every
	deletion from agent~1's bundle leaves value $16$. This ratio is close
	to one. The example establishes the failure of exact EF1 with very
	few goods, but leaves open how much envy an equilibrium can sustain.

	We now improve this upper bound to $2/3$.
	The construction has two principal agents, agent~1 and agent~2.
	Agent~2 is the agent whose envy we examine; she receives goods that
	she values at $2$, whereas agent~1 receives goods that agent~2 values
	at $3$. Agent~1 picks before agent~2 in every round and ranks her own
	goods first. So, if agent~2 takes some of these goods, then agent~1
	runs out of them earlier and starts taking the goods of agent~2.
	At the same time, each of the remaining agents, whom we call
	\emph{auxiliary}, removes a specific good of agent~2 by a
	prescribed deadline, unless agent~2 has already taken it. Together,
	these two effects force agent~2 to lose at least three of her own
	goods for every two goods of agent~1 that she takes. Therefore, no
	deviation is profitable for her.

	The auxiliary agents implement the deadlines through fixed reports,
	while goods that only they value fill their remaining turns without
	disturbing the intended allocation. Giving unit valuations to all
	agents other than agent~2 makes their equilibrium conditions
	immediate. This leaves a single best-response condition to verify,
	namely that of agent~2.

	\begin{theorem}\label{thm:upper-bound}
		For every integer $k\ge4$, there is an additive instance with $k+1$
		agents and $k(k+1)$ goods, and a pure Nash equilibrium of Round-Robin
		whose allocation is not $\alpha$-EF1 for any
		\[
		\alpha\in\Big(\frac{2k}{3(k-1)},1\Big].
		\]
		Consequently, it is not possible to guarantee $\alpha$-EF1 for every pure Nash equilibrium of Round-Robin for any constant  $\alpha >2/3$.
	\end{theorem}
	
	\begin{proof}
		We begin by describing the instance. There are $n=k+1$ agents: the
		principal agents $1$ and $2$, and the auxiliary agents
		$3,\dots,k+1$, who pick in this order in every round. The set of goods is
		$M=H\cup G\cup Z_1\cup\dots\cup Z_{k-1}$, where
		\[
		H=\{h_1,\dots,h_k\},\qquad G=\{g_1,\dots,g_k\},\qquad
		Z_t=\{z_{t,1},\dots,z_{t,k}\}\ \ \text{for } 1\le t\le k-1,
		\]
		and all these sets are pairwise disjoint. The goods of $Z_t$ are meant
		for agent $t+2$. We have $m=2k+(k-1)k=kn$, so there are exactly $k$
		complete rounds and every agent receives exactly $k$ goods, regardless
		of the reports. The valuation function of agent~2 is additive with
		\[
		v_2(g)=\begin{cases}
			2,&g\in H,\\
			3,&g\in G,\\
			0,&\text{otherwise,}
		\end{cases}
		\]
		whereas every other agent values every good at $1$.
		
		For the reports of the agents, we first fix an arbitrary strict
		ranking of $M$. Every ranking displayed below is completed by
		appending all unlisted goods to its right, in this fixed order.
		The reports are
		\[
		\begin{aligned}
			\widehat{\succ}_1&:\quad
			g_1\succ\cdots\succ g_k\succ h_1\succ\cdots\succ h_k\succ\cdots,\\
			\widehat{\succ}_2&:\quad
			h_1\succ\cdots\succ h_k\succ\cdots,\\
			\widehat{\succ}_{t+2}&:\quad
			z_{t,1}\succ\cdots\succ z_{t,t-1}\succ h_t
			\succ z_{t,t}\succ\cdots\succ z_{t,k}\succ\cdots,
			\qquad 1\le t\le k-1,
		\end{aligned}
		\]
		where, for $t=1$, no good of $Z_1$ precedes $h_1$ in agent~3's ranking. These are strict
		rankings and, as discussed in Section~\ref{sec:setting}, they can be
		implemented by bid vectors.\medskip
		
		We first argue about the resulting allocation. 
		We claim that, in every round $r$, agent~1 takes $g_r$, agent~2 takes
		$h_r$, and agent $t+2$ takes $z_{t,r}$, for every $1\le t\le k-1$.
		We use induction on $r$. Suppose that the claim holds for all rounds
		before round $r$; for $r=1$ this assumption is vacuous. Then the
		goods allocated before round $r$ are exactly $g_s$, $h_s$, and
		$z_{t,s}$, for $s<r$ and $1\le t\le k-1$. In particular, $g_r$ and
		$h_r$ are available at the beginning of round $r$, whereas all goods
		that precede them in $\widehat{\succ}_1$ and $\widehat{\succ}_2$,
		respectively, have been allocated. Hence agent~1 takes $g_r$ and
		agent~2 takes $h_r$. Now consider agent $t+2$. Since the agents before
		her in round $r$ have taken $g_r$, $h_r$, and goods of
		$Z_1,\dots,Z_{t-1}$, the goods of $Z_t$ that are not available at her
		turn are exactly $z_{t,1},\dots,z_{t,r-1}$; in particular, $z_{t,r}$ is
		available. If $r<t$, the goods that she ranks above $z_{t,r}$ are
		$z_{t,1},\dots,z_{t,r-1}$. If $r\ge t$, these are
		$z_{t,1},\dots,z_{t,r-1}$ and $h_t$, where $h_t$ has been taken by
		agent~2 in round $t$. In both cases, all goods that she ranks above
		$z_{t,r}$ are unavailable, so she takes $z_{t,r}$. This completes the
		induction. Therefore,
		\[
		A_1=G,\qquad A_2=H,\qquad A_{t+2}=Z_t\ \ \text{for } 1\le t\le k-1,
		\]
		and, in particular, $v_2(A_2)=2k$.\medskip
		
		Next, we show that the profile above is a pure Nash equilibrium.
		Every agent other than agent~2 receives $k$ goods of value $1$ to her.
		As every agent receives exactly $k$ goods under every report profile,
		this is the maximum value she can obtain, and thus her report is a
		best response. It remains to show that $\widehat{\succ}_2$ is a best
		response for agent~2. To this end, consider any report of agent~2,
		keep the reports of all other agents fixed, and let
		$D=(D_1,\dots,D_n)$ be the resulting allocation. We are going to show
		that $v_2(D_2)\le 2k$.
		
		We begin by showing that, for every $1\le t\le k-1$, the good $h_t$ is
		allocated by the end of round $t$. If $h_t$ is not available at the
		turn of agent $t+2$ in round $t$, then it has already been allocated.
		So, suppose that $h_t$ is available at that turn. Then $h_t$ was also
		available at each of her $t-1$ earlier turns, so at each of these
		turns she took a good that she ranks above $h_t$, i.e., one of
		$z_{t,1},\dots,z_{t,t-1}$. These $t-1$ distinct picks exhaust
		$\{z_{t,1},\dots,z_{t,t-1}\}$, and therefore $h_t$ is the
		highest-ranked available good of agent $t+2$ in round $t$; hence, she takes it. Note that
		this argument applies to every report of agent~2, including reports
		under which she takes goods of $Z_t$ in order to change the later
		picks of agent $t+2$.

		If $|D_2\cap G|=0$, then $v_2(D_2)=2|D_2\cap H|\le 2k$.
		So, assume that $|D_2\cap G|\ge1$. 
		As agent~1 ranks
		the goods of $G$ first, at each of her turns she takes a good of $G$
		if and only if such a good is available. Since the set of available
		goods only shrinks, the rounds in which a good of $G$ is available at
		agent~1's turn are the first few rounds. As agent~1 has one turn per
		round and receives $|D_1\cap G|$ goods of $G$, these are exactly the rounds
		$1,\dots,|D_1\cap G|$. Moreover, as agent~1 picks before agent~2 in every round,
		if agent~2 takes a good of $G$ in some round $r$, then this good was
		available at agent~1's turn in round $r$, and thus $r\le |D_1\cap G|$. Since agents~1 and~2 receive distinct
		goods of $G$, we have $|D_1\cap G|+|D_2\cap G|\le k$; in particular, $|D_1\cap G|\le k-1$.
		
		We now bound $|D_2\cap H|$. In rounds $1,\dots,|D_1\cap G|$, agent~2 has
		exactly $|D_1\cap G|$ turns, and, by the above, she takes all $|D_2\cap G|$ goods of
		$D_2\cap G$ at these turns. So, she receives at most $|D_1\cap G|-|D_2\cap G|$ goods of
		$H$ in these rounds. Next, consider rounds
		$|D_1\cap G|+1,\dots,k$. As $|D_1\cap G|\le k-1$, the goods $h_1,\dots,h_{|D_1\cap G|}$ have been
		allocated by the end of round $|D_1\cap G|$, and thus every good of $H$
		allocated in these rounds belongs to $\{h_{|D_1\cap G|+1},\dots,h_k\}$. Suppose
		that agent~2 takes some $h\in H$ in a round $r\ge |D_1\cap G|+1$. Then $h$ was
		available at agent~1's turn in round $r$, whereas no good of $G$ was,
		since $r>|D_1\cap G|$. As agent~1 ranks the goods of $H$ right after those of
		$G$, she took a good of $H$ at that turn, which is different from
		$h$, since $h$ was still available afterwards. That is, in every round
		$r\ge |D_1\cap G|+1$ in which agent~2 receives a good of $H$, agent~1 also
		receives a good of $H$. As every agent receives one good per round, in
		rounds $|D_1\cap G|+1,\dots,k$ agent~1 receives at least as many goods of $H$ as
		agent~2 does. All these goods are distinct and belong to
		$\{h_{|D_1\cap G|+1},\dots,h_k\}$, so agent~2 receives at most $(k-|D_1\cap G|)/2$ goods
		of $H$ in these rounds. Combining the above, we get
		\begin{align*}
			v_2(D_2) &=3|D_2\cap G|+2|D_2\cap H|\\
			&\le 3|D_2\cap G|+2(|D_1\cap G|-|D_2\cap G|)+(k-|D_1\cap G|)\\
			&=k+|D_1\cap G|+|D_2\cap G|\le 2k,
		\end{align*}
		where the last inequality follows from $|D_1\cap G|+|D_2\cap G|\le k$.
		
		Since $v_2(A_2)=2k$, the report $\widehat{\succ}_2$ is a best response
		for agent~2. Together with the best responses of the other agents,
		this shows that the profile is a pure Nash equilibrium.\medskip

		Finally, we show that the allocation $(A_1,\ldots,A_n)$ fails to be $\alpha$-EF1 for sufficiently large $\alpha$. Every good of $A_1=G$ is worth $3$ to agent~2. Thus, for every
		$g\in A_1$, we have $v_2(A_1\setminus\{g\})=3(k-1)$, while
		$v_2(A_2)=2k$. That is, the allocation is not $\alpha$-EF1 for any
		$\alpha>2k/(3(k-1))$, and, since $k\ge4$, we have $2k/(3(k-1))<1$.
		Note that
		\[
		\frac{2k}{3(k-1)}=\frac23+\frac{2}{3(k-1)}.
		\]
		So, given any $\alpha\in(2/3,1]$, we may choose an integer $k\ge4$ so that
		$2k/(3(k-1))<\alpha$, and the corresponding pure Nash equilibrium is
		not $\alpha$-EF1.
	\end{proof}
	
	Note that the upper bound also holds for strictly positive valuations. For each
	$k$, we may increase the value of agent~2 for each good by $1/k$, leaving
	the reports and all other valuations unchanged. Since agent~2 receives
	exactly $k$ goods under every report profile, the value of every
	bundle she may receive increases by exactly $1$, and thus the profile
	remains an equilibrium. The EF1 comparison now gives
	\[
	\frac{2k+1}{(3+1/k)(k-1)}\longrightarrow\frac23,
	\]
	as $k\to\infty$. Hence, even strictly positive instances rule out every factor strictly greater than $2/3$.

	\section{Discussion and open problems}\label{sec:discussion}
	
	Our results leave a gap between the guaranteed factor of $1/2$ and the
	upper bound of $2/3$ established here with respect to EF1. To state the question precisely,
	let $\alpha_{\mathrm{RR}}$ be the largest factor guaranteed at every
	pure Nash equilibrium of every Round-Robin instance, and let $\alpha_{\mathrm{RB}}$
	be the corresponding factor for all fixed recursively balanced
	sequences. Since Round-Robin belongs to the latter class,
	\[
	\frac12\le\alpha_{\mathrm{RB}}
	\le\alpha_{\mathrm{RR}}\le\frac23.
	\]
	Our central open question is to determine these factors. In particular,
	does allowing the order to change between rounds affect the sharp
	guarantee? The lower-bound proof uses only that no agent is more
	than one turn ahead of another, whereas the upper-bound construction
	uses the same order in every round. 
	
	The dependence on the number of agents is another natural question.
	For two agents, the quantitative proportionality bound gives exact
	EF1 (Corollary~\ref{cor:two-agents}). The small counterexample shows that this fails already for
	three agents, but the family approaching $2/3$ uses a growing number
	of auxiliary agents. Determining the best guaranteed factor for each fixed
	number of agents could distinguish which part of the gap comes
	from additional strategic interaction and which part comes from
	the number of rounds.
	
	There is also a useful distinction between equilibrium existence
	and a guarantee for every equilibrium. The bluff profile provides
	an EF1 equilibrium, while the examples show that other equilibria
	can sustain envy after every deletion. It would be interesting to
	identify conditions on reports or valuations that force all
	equilibria to satisfy a stronger guarantee, without restricting
	the agents to truthful reports.

\section*{Acknowledgments}
Rebecca Reiffenh{\"a}user was partially supported by the NWO grant \emph{Fair Allocation for Strategic Agents (FAStA)}, grant ID https://doi.org/10.61686/AZZOJ38639.
	
	\bibliographystyle{plainnat}
	\bibliography{abb,RR-EF1}

@article{KoCh71a,
	author = {D. A. Kohler and R. Chandrasekaran},
	journal = {Operations Research},
	number = {2},
	pages = {270--277},
	title = {A Class of Sequential Games},
	volume = {19},
	year = {1971}}

@inproceedings{ABLM17a,
  author       = {Haris Aziz and
                  Sylvain Bouveret and
                  J{\'{e}}r{\^{o}}me Lang and
                  Simon Mackenzie},
  xeditor       = {Satinder Singh and
                  Shaul Markovitch},
  title        = {Complexity of Manipulating Sequential Allocation},
  booktitle    = {Proceedings of the Thirty-First {AAAI} Conference on Artificial Intelligence, {AAAI} 2017},
  pages        = {328--334},
  publisher    = {{AAAI} Press},
  year         = {2017},
}

@inproceedings{AWX15b,
	author = {H. Aziz and T. Walsh and L. Xia},
	booktitle = proc # {24th} # ijcai,
	pages = {468--474},
	title = {Possible and Necessary Allocations via Sequential Mechanisms},
	year = {2015}}

@inproceedings{BoLa11a,
	author = {S. Bouveret and J. Lang},
	booktitle = proc # {22nd} # ijcai,
	pages = {73--78},
	publisher = {AAAI Press},
	title = {A General Elicitation-Free Protocol for Allocating Indivisible Goods},
	year = {2011}}

@inproceedings{ABFLLR21,
  author    = {Amanatidis, Georgios and Birmpas, Georgios and Fusco, Federico and
               Lazos, Philip and Leonardi, Stefano and Reiffenh{\"a}user, Rebecca},
  title     = {Allocating Indivisible Goods to Strategic Agents: Pure {N}ash
               Equilibria and Fairness},
  booktitle = {Proceedings of the 17th International Conference on Web and
               Internet Economics (WINE)},
  series    = {Lecture Notes in Computer Science},
  volume    = {13112},
  pages     = {149--166},
  publisher = {Springer},
  year      = {2021}
}

@article{ABFLLR24,
  author  = {Amanatidis, Georgios and Birmpas, Georgios and Fusco, Federico and
             Lazos, Philip and Leonardi, Stefano and Reiffenh{\"a}user, Rebecca},
  title   = {Allocating Indivisible Goods to Strategic Agents: Pure {N}ash
             Equilibria and Fairness},
  journal = {Mathematics of Operations Research},
  volume  = {49},
  number  = {4},
  pages   = {2425--2445},
  year    = {2024},
}

@article{CKMPSW19,
  author  = {Caragiannis, Ioannis and Kurokawa, David and Moulin, Herv{\'e} and
             Procaccia, Ariel D. and Shah, Nisarg and Wang, Junxing},
  title   = {The Unreasonable Fairness of Maximum {N}ash Welfare},
  journal = {ACM Transactions on Economics and Computation},
  volume  = {7},
  number  = {3},
  pages   = {12:1--12:32},
  year    = {2019}
}

@article{ABLLR25,
  author       = {Georgios Amanatidis and
                  Georgios Birmpas and
                  Philip Lazos and
                  Stefano Leonardi and
                  Rebecca Reiffenh{\"{a}}user},
  title        = {Round-Robin Beyond Additive Agents: Existence and Fairness of Approximate
                  Equilibria},
  journal      = {Mathematics of Operations Research},
  volume       = {50},
  number       = {4},
  pages        = {2717--2737},
  year         = {2025}
}

@inproceedings{AGW17,
  author       = {Haris Aziz and
                  Paul Goldberg and
                  Toby Walsh},
  xeditor       = {J{\"{o}}rg Rothe},
  title        = {Equilibria in Sequential Allocation},
  booktitle    = {Proceedings of the 5th International Conference on Algorithmic Decision Theory, {ADT}
                  2017},
  series       = {Lecture Notes in Computer Science},
  volume       = {10576},
  pages        = {270--283},
  publisher    = {Springer},
  year         = {2017}
}

@STRING{proc = {Proceedings of the }}

@STRING{ijcai = { International Joint Conference on Artificial Intelligence (IJCAI)}}

@STRING{aaai = { AAAI Conference on Artificial Intelligence (AAAI)}}

@STRING{wine = { International Workshop on Internet and Network Economics (WINE)}}

@STRING{adt = { International Conference on Algorithmic Decision Theory (ADT)}}

@STRING{springer = {Springer-Verlag}}

@STRING{acm = {ACM Press}}

\end{document}